\documentclass[11pt]{article}

\usepackage{silence}
\ActivateWarningFilters[pdftoc]
\ActivateWarningFilters[etexwarning]

\usepackage{amsmath,amssymb,amsthm}
\usepackage[round]{natbib}
\usepackage{mathtools}
\usepackage{comment}

\usepackage[margin=1.25in]{geometry}
\usepackage{fancyhdr, titling}
\usepackage{times}

\usepackage{hyperref}
\usepackage{cleveref}
\usepackage{subcaption}

\usepackage[ruled]{algorithm2e}
\usepackage{algorithmic}

\newtheorem{proposition}{Proposition}
\newtheorem{theorem}{Theorem}

\theoremstyle{remark}

\newtheorem{remark}{Remark}
\newtheorem{example}{Example}

\newcommand{\reals}{\mathbb{R}}

\newcommand{\hypset}{\mathcal{H}}
\newcommand{\rejset}{R}
\newcommand{\Rcal}{\mathcal{R}}

\newcommand{\Cclosed}{\overline{\mathcal{C}}}

\newcommand{\FDP}{\textnormal{FDP}}
\newcommand{\FDR}{\textnormal{FDR}}

\newcommand{\eBH}{\textnormal{eBH}}
\newcommand{\BY}{\textnormal{BY}}

\newcommand{\textcBY}{\overline{\text{BY}}}
\newcommand{\cBY}{\overline{\textnormal{BY}}}

\newcommand{\SC}{\textnormal{SC}}
\newcommand{\CeBH}{{\overline{\textnormal{eBH}}}}
\newcommand{\textCeBH}{{\overline{\text{eBH}}}}

\newcommand{\ind}{\mathbf{1}}
\newcommand{\F}{\text{F}}
\newcommand{\expect}{\mathbb{E}}

\newcommand{\Hcal}{\mathcal{H}}

\usepackage{autonum}

\title{Bringing closure to FDR control:\\ beating the e-Benjamini-Hochberg procedure}
\author{
Ziyu Xu\thanks{Department of Statistics and Data Science, Carnegie Mellon University, USA. Email: \texttt{xzy@cmu.edu}.}
\and
Lasse Fischer\thanks{Competence Center for Clinical Trials Bremen, University of Bremen, Germany. Email: \texttt{fischer1@uni-bremen.de}.}
\and
Aaditya Ramdas\thanks{Department of Statistics and Data Science and Machine Learning Department, Carnegie Mellon University, USA. Email: \texttt{aramdas@cmu.edu}.}
}

\begin{document}
\footnotetext{This paper is subsumed by the merged work of \citet{xu_bringing_closure_2025a}.\newline}

\maketitle

\begin{abstract}
    False discovery rate (FDR) has been a key metric for error control in multiple hypothesis testing, and many methods have developed for FDR control across a diverse cross-section of settings and applications. We develop a closure principle for all FDR controlling procedures, i.e., we provide a characterization based on e-values for all admissible FDR controlling procedures. A general version of this closure principle can recover any multiple testing error metric and allows one to choose the error metric post-hoc. We leverage this idea to formulate the closed eBH procedure, a (usually strict) improvement over the eBH procedure for FDR control when provided with e-values. This also yields a closed BY ($\textcBY$) procedure that dominates the Benjamini-Yekutieli (BY) procedure for FDR control with arbitrarily dependent p-values, thus proving that the latter is inadmissibile. We demonstrate the practical performance of our new procedures in simulations.
\end{abstract}

\tableofcontents

\section{Introduction}

Multiple testing with p-values, in particular false discovery rate (FDR) control, has been a centerpiece of statistical research in the past three decades since its introduction by \citet{benjamini_controlling_false_1995}.
Recently, \citet{wang_false_discovery_2022} formulated the e-Benjamini-Hochberg (eBH) procedure for FDR control, which uses e-values for each hypothesis and is robust to arbitrary dependence. It is known that by employing external randomization, or additional information about distribution of the e-values or their dependence structure, the eBH procedure can be improved (i.e.\ it can be amended to be more powerful).

In this paper, we will show that one can (often strictly) improve the eBH procedure with no additional information or randomization. We call our procedure the \emph{closed eBH} procedure, denoted $\textCeBH$ for brevity. Further, we establish that \text{any} FDR controlling procedure \emph{must} rely on e-values, and in particular can be written as an instance of $\textCeBH$ acting on some e-values. In the process, we establish a closure principle for FDR control, similar to the results of \citet{marcus_closed_testing_1976} and \citet{sonnemann_vollstaendigkeitssaetze_fuer_1988} for FWER control. This is also complementary to the results of~\cite{ignatiadis_asymptotic_compound_2025} who show that every FDR controlling procedure is an instance of the eBH procedure with \emph{compound} e-values, a fact that we return to later in this work.

\paragraph{Problem setup.} The standard multiple testing problem assumes that there are $K$ null hypotheses $\hypset_1,\dots,\hypset_K$ a scientist wishes to test, some of which are truly null, and one wishes to discover the non-nulls.
Let $\alpha \in [0, 1]$ be our fixed and desired level of control on the FDR.

First, for any set $A \subseteq [K] \coloneqq \{1,\dots,K\}$, define $\hypset_A \coloneqq \cap_{j \in A} \hypset_j$, i.e., the intersection hypotheses of the base hypotheses corresponding to the set $A$.
For any set $A$, if the true null set equals $A$, then
 the false discovery proportion (FDP) of a candidate discovery set $\rejset \subseteq [K]$ equals
\[
\text{FDP}_A(\rejset) = \sum_{j \in A} \frac{\ind\{j \in \rejset\}}{|\rejset| \vee 1} = \frac{|\rejset \cap A|}{|\rejset| \vee 1}.
\]
In words, if $\Hcal_A$ happened to be true (and $\Hcal_B$ is false for every $B \supset A$), and $\rejset$ was rejected, then $\text{FDP}_A(\rejset)$ would be the realized FDP.

Let $\mathbf{X} \in \mathcal{X}$ denote the data, which is sampled from the true distribution $\mathcal{P}^*$.
We define a \emph{discovery procedure} $\Rcal: \mathcal{X} \rightarrow 2^{[K]}$ as a mapping from input data to a discovery set $R \subseteq [K]$, where $\reals^{\geq 0}$ represents the nonnegative reals.
Formally, each base hypothesis $\hypset_j$ is a set of distributions, and $\hypset_j$ is said to be ``truly null'' if $\mathcal P^* \in \hypset_j$. Let $A^* \subseteq [K]$ denote the set of true nulls, i.e., $\mathcal{P}^* \in \hypset_{A^*}$ and $\mathcal{P}^* \not\in \hypset_A$ for all $A$ s.t.\ $A^* \subset A$.
The FDR of a discovery procedure $\Rcal$ equals $$\FDR(\Rcal) \coloneqq \expect[\FDP_{A^*}(\Rcal(\mathbf{X}))].$$

\paragraph{The eBH procedure.}
We consider the situation where our data has already been summarized into $K$ e-values $\mathbf{X} = \mathbf{E} \coloneqq (E_1,\dots,E_K)$ (one per hypothesis). These e-values may be arbitrarily dependent.
\citet{wang_false_discovery_2022} introduced the eBH procedure which takes in ensures FDR control below $\alpha$ under arbitrary dependence among the e-values, and is defined as producing the following discovery set:
\begin{align}
    k^\eBH &\coloneqq \max\ \left\{k \in \{0\} \cup[K]: \sum_{i \in [K]} \ind\left\{E_i \geq \frac{K}{\alpha k} \right\}\geq k\right\}\\
\rejset^{\eBH} &\coloneqq \Rcal^\eBH(\mathbf{E}) \coloneqq  \left\{i \in [K]: E_i \geq \frac{K}{\alpha k^\eBH}\right\}
\end{align}
The eBH procedure (and variants) has been used to derive FDR control in many settings with complex dependencies in the data, e.g., bandit data collection \citep{xu_unified_framework_2021}, multi-stream monitoring \citep{wang_anytime-valid_fdr_2025,dandapanthula_multiple_testing_2025,tavyrikov_carefree_multiple_2025}, multi-resolution testing \citep{gablenz_catch_me_2025,yu_generalized_e-value_2025}, combining data from multiple sources \citep{banerjee_harnessing_collective_2024,li_note_e-values_2025}, etc. Prior work has shown that existing FDR controlling methodologies can be recast as an application of the eBH procedure. Examples include knockoffs \citep{ren_derandomised_knockoffs_2023,guan_one-at-a-time_knockoffs_2025a} as well as the seminal procedure of \citet{benjamini_controlling_false_1995} \citep{li_note_e-values_2025} and its variant introduced by \citet{benjamini_control_false_2001} that is robust to arbitrary dependence \citep{wang_false_discovery_2022}. We will show that the eBH procedure is actually inadmissible, and derive a procedure that (usually strictly) improves upon eBH that is applicable to and may improve power for all of the aforementioned situations.

\begin{comment}
\paragraph{Our contributions.} We make the following contributions in this paper:
\begin{itemize}
   \item \textit{Closure principle for FDR and general multiple testing error metrics:} We establish a closure principle for FDR control that provides a complete characterization of all admissible FDR controlling procedures using e-values. We extend this characterization to many multiple testing error metrics.

   \item \textit{The $\textCeBH$ and $\textcBY$ procedures:} We introduce the closed eBH ($\textCeBH$) procedure, which provably dominates the standard eBH procedure while maintaining the same FDR guarantees under arbitrary dependence. We demonstrate that $\mathcal{C}^\text{SC} \subseteq \overline{\mathcal{C}}$, which implies $R^\text{eBH} \subseteq R^{\textCeBH}$.

   \item \textit{BY inadmissibility:} We prove the inadmissibility of the BY procedure by constructing the closed BY ($\textcBY$) procedure that strictly dominates it. This is achieved by applying our closure principle to the e-values derived from p-values via the calibrator $f^{\text{BY}(\alpha, K)}$.

   \item \textit{Efficient algorithm:} We develop an efficient $O(K^3)$ dynamic programming algorithm for computing the $\textCeBH$ procedure, making it practical for applications with large numbers of hypotheses.

\end{itemize}
\end{comment}

\paragraph{Contributions.}
Our contributions are listed below.
\begin{itemize}
    \item Development of a general closure principle that encompasses all expectation based error rates and allows to construct all procedures controlling these metrics (Theorem~\ref{theo:CP_general}). This characterizes multiple testing procedures by single e-values for the intersection hypotheses and often yields improvements of existing procedures.  While closure principles for FWER control \citep{marcus_closed_testing_1976} and FDP probability bounds \citep{goeman2011multiple} are known, our principle additionally includes FDR control which was an open problem, while also recovering these earlier known closure principles.
    \item Illustration of the application of our closure principle by introducing a closed eBH ($\textCeBH$) and a closed BY ($\textcBY$) procedure that dominate the original methods (Section~\ref{sec:closed-ebh}). Both eBH and BY are well-studied and were not generally known to be inadmissible (it was known that BY can be improved by randomization, but not by a deterministic procedure \citep{xu_more_powerful_2023}).
    \item Showing that in its most general form our closure principle allows for post-hoc choice (in a specific sense) of the significance level $\alpha$  (Section~\ref{sec:closure}) and for post-hoc choice of the error metric (Section~\ref{sec:general-closure}), thus providing multiple testing guarantee that has not been considered before.
\end{itemize}

\paragraph{Paper outline.} We introduce the $\textCeBH$ procedure in \Cref{sec:closed-ebh}. In \Cref{sec:compound}, we formulate the closed compound eBH procedure for compound e-values. Then, we formulate our closure principle for all FDR controlling procedures in \Cref{sec:closure} and then extend it to any expectation based error rate in \Cref{sec:general-closure}. We describe some extensions and synergies of the $\textCeBH$ procedure with existing e-value techniques in  \Cref{sec:extensions} (i.e., randomization, e-merging, boosting). We present simulations showing the improvements of $\textCeBH$ over existing eBH methodology in \Cref{sec:simulations}. Lastly, we discuss prior work in depth in \Cref{sec:related-work} and recap our findings in \Cref{sec:conclusion}.

\section{The closed eBH ($\textCeBH$) procedure}\label{sec:closed-ebh}

Our procedure rejects more hypotheses (often strictly) than eBH, while still controlling the FDR.
The construction proceeds akin to the closure principle for familywise error control \citep{marcus_closed_testing_1976} by first defining a test for every intersection hypothesis.
Define the e-value to test $\Hcal_A$ as
\[
    E_A =|A|^{-1} \sum_{i \in A} E_i . \label{eq:avg-e-value}
\]
The arithmetic mean is a symmetric admissible e-value merging function \citep{vovk_e-values_calibration_2021}. Let $\expect_A$ denote the supremum of the expectations of a random variable over distributions where $\Hcal_A$ holds (i.e., for any distribution $\mathcal{P} \in \hypset_A$). By the definition of $E_j$ being an e-value for each $j \in A$, for any $A \subseteq [K]$, we have that
\[
\expect_A[E_A] \leq 1.
\]
We define the set of \emph{candidate} discovery sets as
\begin{align}
\label{eq:candidates}
\overline{\mathcal{C}}\coloneqq \left\{\rejset \subseteq [K]:
E_A \geq \frac{\FDP_A(\rejset)}{\alpha} \quad \forall {A \subseteq [K]}
\right\}.
\end{align}

\begin{proposition}\label{prop:cebh_FDR}
    For any procedure $\Rcal$ that ensures $\rejset \in \overline{\mathcal{C}}$, we have that $\FDR(\Rcal) \leq \alpha$.
\end{proposition}
\begin{proof}
For any $\rejset \in \overline{\mathcal{C}}$, we have that
    \begin{align}
        \FDP_{A^*}(\rejset) \leq \alpha E_{A^*}.
    \end{align}
    Taking expectations, the claim is proved.
\end{proof}
Let $R_i \subseteq [K]$ be the set of hypotheses corresponding to the $i$ largest e-values (with ties being arbitrarily broken) for each $i \in [K]$.
Define the procedure that outputs largest discovery set in $\overline{\mathcal{C}}$ as the $\textCeBH$ procedure:
\begin{align}
\label{eq:k-cebh}
    k^\CeBH &\coloneqq \max\ \{i \in [K]: R_i \in \overline{\mathcal{C}}\}\\
\label{eq:R-cebh}
    \rejset^\CeBH &\coloneqq \Rcal^\CeBH(\mathbf{E}) \coloneqq R_{[:k^\CeBH]}.
\end{align}
Note that $\rejset^\CeBH \in \underset{\rejset \subseteq \overline{\mathcal{C}}}{\textnormal{argmax}} |\rejset|$ by definition.
We will now show that this method is provably (usually strictly) more powerful as eBH, the standard way to control the false discovery rate (FDR) in this setting. Another way of interpreting is eBH is to first define \emph{self-consistent} candidate discovery sets $\mathcal{C}^\SC$, and then note that $\rejset^\eBH$ is the largest such set. Formally,
\begin{align}
    \mathcal{C}^{\SC} &\coloneqq \left\{\rejset \subseteq [K]: \min_{i \in \rejset}\ E_i \geq K / (\alpha |\rejset|) \right\},\\
    \rejset^{\eBH} &= \underset{\rejset \subseteq \mathcal{C}^\SC}{\textnormal{argmax}}\  |\rejset|.
\end{align}

\begin{proposition} \label{prop:cebh-dominates-ebh}
    $\mathcal{C}^\SC \subseteq \overline{\mathcal{C}}$, and consequently $\rejset^\eBH \subseteq \rejset^{\CeBH}$.
\end{proposition}
\begin{proof}
    Consider the following derivation for arbitrary $A$ and $\rejset$:
    \begin{align}
        E_A &= |A|^{-1}\sum\limits_{i \in A} E_i\\
            &\geq
            |A|^{-1} |A\cap \rejset| \min_{i \in A \cap \rejset} E_i
            \geq
            |A|^{-1}|A\cap \rejset|\min_{i \in \rejset} E_i
            \geq
            K^{-1} |A\cap \rejset|\min_{i \in \rejset} E_i.
        \label{eq:lb-ineq}
    \end{align}
    Recall that $E_A$ is defined as the arithmetic mean e-value in \eqref{eq:avg-e-value}. The first inequality is lower bounding each term in the sum and also dropping some positive terms. The second inequality follows by strictly expanding the set which we take a minimum over. The third inequality is because $K \geq |A|$.

    Now, assume that $\FDP_A(\rejset) > 0$ and $\rejset \in \mathcal{C}^\SC$ (otherwise $\FDP_A(\rejset) / E_A \leq \alpha$ is trivially satisfied).
    We have the following result for any $\rejset \in \mathcal{C}^\SC$:
    \begin{align}
    E_A
    \geq
     K^{-1} |A\cap \rejset|\min_{i \in \rejset} E_i
    \geq
    \frac{|A \cap \rejset|}{\alpha |\rejset|}
    =
    \frac{\FDP_A(\rejset)}{\alpha}.\label{eq:ineq-avg-lb}
    \end{align}
    The first inequality is by \eqref{eq:lb-ineq}, and the second inequality is by the definition of $\mathcal{C}^\SC$.
Thus, we get that
    \begin{align}
        \frac{\FDP_A(\rejset)}{E_A} \leq \alpha
    \end{align} for any $\rejset \in \mathcal{C}^\SC$ by \eqref{eq:ineq-avg-lb}. Thus, we get that $\mathcal{C}^\SC \subseteq \overline{\mathcal{C}}$ and also that $R^\eBH \in \mathcal{C}$ Since $R^\CeBH \in \underset{R \in \overline{\mathcal{C}}}{\text{argmax}}\ |R|$, we have that $|R^\eBH| = k^\eBH \leq k^{\CeBH} = |R^\CeBH|$ and consequently $R^\eBH = R_{[:k^\eBH]} \subseteq R_{[:k^\CeBH]} = R^\CeBH$. Thus, we have shown our desired results.
\end{proof}

The looseness in the inequalities in \eqref{eq:lb-ineq} hints at $\rejset^\CeBH$ being potentially more powerful than eBH, and we will observe that is indeed the case in our simulations.

\citet[Section S3.2]{ignatiadis_e-values_unnormalized_2022} formulates the \emph{minimally adaptive eBH procedure} where one first tests the global null with the arithmetic mean of the e-values, and only once the global null rejects, the eBH procedure is applied at level $\alpha K / (K - 1)$. Formally we can define
\begin{align}
k^{\eBH\text{m}} &\coloneqq \max \left\{k \in [K]: E_{[K]} \geq \alpha^{-1} \text{ and }\sum_{i = 1}^K \ind\left\{E_i \geq \frac{K - 1}{\alpha k}\right\} \geq k\right\} \cup \{0\},\\
\rejset^{\eBH\text{m}} &\coloneqq \Rcal^{\eBH\text{m}}(\mathbf{E}) \coloneqq R_{[:k^{\eBH\text{m}}]}
.
\end{align}They show that this method dominates eBH as well.
This method is, in turn, dominated by the $\textCeBH$ procedure.
\begin{proposition} \label{prop:cebh-dominates-ebhm}
$\rejset^{\eBH\textnormal{m}} \subseteq \rejset^{\CeBH}$.
\end{proposition}
\begin{proof}
    By definition of $R^{\eBH\text{m}}$, we know that if any discoveries are made, $$E_{[K]} \geq \alpha^{-1} = \frac{\FDP_{[K]}(R^{\eBH\text{m}})}{\alpha}.$$
    Consider the following derivation for arbitrary $A \subset [K]$, i.e, $|A| \leq K - 1$, and $\rejset\subseteq [K]$:
    \begin{align}
      E_A &= |A|^{-1}\sum\limits_{i \in A} E_i\\
            &\geq
            |A|^{-1} |A\cap \rejset| \min_{i \in A \cap \rejset} E_i
            \geq
            |A|^{-1}|A\cap \rejset|\min_{i \in \rejset} E_i
            \geq
            (K - 1)^{-1} |A\cap \rejset|\min_{i \in \rejset} E_i.
        \label{eq:lb-ineq-m}
    \end{align}
    Recall that $E_A$ is defined as the arithmetic mean e-value in \eqref{eq:avg-e-value}. The first inequality is lower bounding each term in the sum and also dropping some positive terms. The second inequality follows by strictly expanding the set which we take a minimum over. The third inequality is because $|A| \leq K - 1$.

    Now, consider $A$ where $\FDP_A(\rejset^{\eBH\text{m}}) > 0$ (otherwise $\FDP_A(\rejset^{\eBH\text{m}}) /\alpha  \leq E_A$ is trivially satisfied).
    We have the following result:
    \begin{align}
    E_A
    \geq
     (K - 1)^{-1} |A\cap \rejset|\min_{i \in \rejset^{\eBH\text{m}}} E_i
    \geq
    \frac{|A \cap \rejset^{\eBH\text{m}}|}{\alpha |\rejset^{\eBH\text{m}}|}
    =
    \frac{\FDP_A(\rejset^{\eBH\text{m}})}{\alpha}.\label{eq:ineq-avg-lb-m}
    \end{align}
    The first inequality is by \eqref{eq:lb-ineq-m}, and the second inequality is by the definition of $\rejset^{\eBH\text{m}}$. Thus, $\rejset^{\eBH\text{m}} \in \overline{\mathcal{C}}$ and we have that $R^{\eBH\text{m}}  \subseteq  R^\CeBH$.
\end{proof}

\paragraph{Where is $\CeBH$ gaining power?} To better illustrate the features of the $\CeBH$ procedure, we can take a close look at the thresholds on the individual e-values for rejection. Choose $A = \{j\}$ for each $j \in R_{[:k]}$, we see that $\overline{\mathcal{C}}$ only enforces the constraint that
\[
E_j \geq \frac{1}{\alpha k^\CeBH} \text{ for every } j \in R^\CeBH.
\]
Of course, it imposes other ``joint constraints'' on averages of these e-values as well: for instance, the average of the smallest $m$ rejected e-values must esceed $m / (\alpha k^{\CeBH})$.
In contrast, the eBH procedure satisfies
\[
E_j \geq \frac{1}{\alpha } \text{ for every } j \in R^\eBH,
\]
which is a significantly harsher threshold.
Meanwhile, the minimally adaptive eBH procedure that makes $k^{\eBH\text{m}} \geq k^\eBH$ rejections satisfies
\[
E_j \geq \frac{1}{\alpha \tfrac{K}{K-1} } \text{ for every } j \in R^{\eBH\text{m}},
\]
which is again a much more stringent condition.
Since $K/(K-1) \leq 2$, with equality only for $K=2$, the latter threshold matches that of the closed eBH procedure for $K=2$, but is more stringent otherwise. It can be easily shown that the closed eBH procedure indeed equals the minimally adaptive eBH procedure for $K=2$, but can be more powerful otherwise, as demonstrated by the following simple example.
\begin{example}[Closed eBH beats minimally adaptive eBH]
    We give a simple example with $K=3$ to show that the closed eBH rule can reject strictly more than the minimally adaptive eBH rule. For the set of input e-values $\mathbf{E} = (60, 39, 11)$, the eBH procedure only makes one rejection at $\alpha = 0.05$, narrowly missing out on the second rejection. The minimally adaptive eBH procedure effectively runs eBH at level $3\alpha/2$, and makes two rejections.
    The closed eBH procedure rejects all 3 hypotheses, which can be checked by verifying that the set $\{1,2,3\}$ lies in $\overline{\mathcal C}$, since each e-value is larger than 20/3, all pairwise averages exceed 40/3 and the overall average exceeds 20.
\end{example}

\begin{algorithm}[h]
    \caption{$O(K^3)$ dynamic programming algorithm for computing discoveries made by the $\textCeBH$ procedure. Let $E_{(i)}$ denote the $i$th smallest e-value for each $i \in [K]$.}
\label{alg:cebh-dp}
\SetAlgoLined
\KwIn{E-values $(E_1, \dots, E_K)$ and target FDR level $\alpha \in [0, 1]$}
\KwOut{Closed eBH discovery set $\rejset^\CeBH$}

\For{$k = K, K - 1, \dots, k^\eBH + 1$}{
  \textit{(Attempt to reject hypotheses with $k$ largest e-values)}\\
  \For{$r = 1$ to $k$}{
    \For{$m = r$ to $r + K - k$}{
        Let $S_{(r, k)} = \sum_{i = 1}^{r} E_{(K - k + i)}$ \quad \textit{(sum of $r$ smallest e-values among top $k$)}.\\
        Let $S_{(m - r)} = \sum_{i = 1}^{m - r} E_{(i)}$ \quad \textit{(sum of $m - r$ smallest e-values overall)}.\\
      Let $E_{(k, r, m)} = \frac{S_{(r, k)} + S_{(m - r)}}{m}$, and $\FDP_{(k, r)} = \frac{r}{k}$.\\
      \If{$\FDP_{(k, r)} / \alpha > E_{(k, r, m)}$}{
        \textbf{continue to next} $k$.
      }
    }
  }
  \Return{$\rejset^\CeBH = R_{[:k]}$ \textit{(discovery set of }$k$\textit{ largest e-values)}}
}
\Return{$\rejset^\CeBH = R^\eBH$}
\end{algorithm}
\paragraph{Efficiently computing $\rejset^\CeBH$}We use a dynamic programming approach to efficiently compute $\rejset^\CeBH$ --- our approach is summarized in Algorithm~\ref{alg:cebh-dp}. First, we note that computing $\rejset^\CeBH$ reduces to computing the discovery set size $k^\CeBH \in [K]$ and setting the discovery set as the hypotheses with the $k^\CeBH$ largest e-values. Recall that the $\textCeBH$ procedure finds the largest rejection set $\rejset \subseteq [K]$ such that the following condition is satisfied for all subsets $A \subseteq [K]$:
\[
    \widehat{\FDR}(\rejset) \coloneqq  \max_{A \subseteq 2^{[K]}}\ \frac{\FDP_A(\rejset)}{E_A} \leq \alpha,
\]
where we let $0/0 = 0$.

The algorithm iteratively considers each $\rejset_k$, which contains the hypotheses corresponding to the $k$ largest e-values, and determines whether $\rejset_k \in \overline{\mathcal{C}}$ is true. The algorithm starts from $K$ and stops as soon as it is able to discover $\rejset_k$ or reaches the eBH discovery set, i.e., at $k = k^\eBH$.

To do this efficiently, we note that for each $r \in [k]$, we can fix $r$ to be the number of false discoveries. Then, to maximize $\FDP_A(\rejset) / E_A$, $A$ includes $r$ true nulls in $\rejset_k$ with the smallest e-values in $\rejset_k$ and we only need to maximize over $m \in \{r, \dots, r + K - k\} $, the number of true nulls outside of the rejection set, where for each $m$ we pick the $m - r$ smallest e-values.

Let $A_{(k, r, m)}$ be the indices corresponding to the $r$ smallest e-values in $\rejset_k$ and the $m - r$ smallest e-values outside of $\rejset_k$. Then, we can say that
$$
\widehat{\FDR}(\rejset_k) = \max_{r \in [k], m \in \{r, r + K - k\}}\ \frac{r / k}{E_{A_{(k, r, m)}}}.
$$

Thus, we ended up with $O(K^3)$ computation, since for each $k \in \{|\rejset^\eBH| + 1, \dots, K\}$, we can compute $\widehat{\FDR}(\rejset_k)$ as a maximization over a $O(K^2)$ sized table which also only take $O(K^2)$ time to compute.

\paragraph{Post-hoc validity}
A stronger form of control on the FDP has been considered in prior literature \citep{grunwald_neyman-pearson_e-values_2024,xu_post-selection_inference_2022,koning_post-hoc_$a$_2024} that the eBH (or any e-self-consistent) procedure satisfies. To define this guarantee, we let $\Cclosed_\beta$ denote $\Cclosed$ as defined in \eqref{eq:candidates} constructed for level $\beta \in [0, 1]$.
\begin{proposition}
    The following post-hoc validity is satisfied for the set of $(\Cclosed_\beta)_{\beta \in [0, 1]}$.
    \begin{align}
        \expect\left[\sup_{\beta \in [0, 1]}\sup_{R \in \Cclosed_\beta}\ \frac{\FDP_{A^*}(R)}{\beta}\right] \leq 1, \label{eq:post-hoc-valid}
    \end{align} where we let $0 / 0 = 1$.
\end{proposition}
\begin{proof}
    By construction of $\Cclosed_\beta$ for each $\beta \in [0, 1]$, we have the following deterministic inequality
    \begin{align}
        \sup_{\beta \in [0, 1]}\sup_{R \in \Cclosed_\beta}\ \frac{\FDP_{A^*}(R)}{\beta} \leq E_{A^*}.
    \end{align} We get our desired result by taking expectations on both sides.
\end{proof}
Thus, we get a form of control even when we choose both $\beta$ and $R$ in a data-dependent fashion.

\subsection{The closed BY ($\textcBY$) procedure\label{sec:BY}} Our results also have implications for the \citet{benjamini_control_false_2001} (BY) procedure, which ensures FDR control under arbitrary dependence among p-values $\mathbf{X} = \mathbf{P} \coloneqq (P_1, \dots, P_K) \in [0, 1]^K$. The procedure is defined as follows:
\begin{align}
    k^{\BY} &\coloneqq \max \left\{k \in \{0\} \cup [K]: \sum_{i = 1}^k \ind\left\{P_i \leq \frac{\alpha k}{K \ell_K} \geq k\right\}\right\}\\
    \rejset^\BY &\coloneqq \Rcal^\BY(\mathbf{P}) \coloneqq R_{(:k^\BY)},
\end{align} where $R_{(:k)}$ denotes the set of hypotheses corresponding to the $k$ smallest p-values.

Define the following p-to-e calibrator $f^{\BY(\alpha, K)}: [0, 1] \rightarrow \reals^{\geq 0}$ \citep{xu_post-selection_inference_2022}:
\begin{align}
    f^{\BY(\alpha, K)}(x) \coloneqq K\alpha^{-1}
    \left(\left\lceil x(\alpha/ (K\ell_K))^{-1} \right\rceil \vee 1 \right)^{-1}\ind\{x \leq \alpha^{-1}\ell_K\}, \label{eq:by-calibrator}
\end{align} where $\ell_K$ is the $K$th harmonic number.
\citet{wang_false_discovery_2022} showed that that this is equivalent to applying eBH to $$\mathbf{E}^{\BY(\alpha, K)} = (f^{\BY(\alpha, K)}(P_1), \dots, f^{\BY(\alpha, K)}(P_K)).\label{eq:by-calib-evalues}$$
Hence, we can now the $\textcBY$ procedure as
\begin{align}
    R^{\cBY} \coloneqq \Rcal^{\cBY} (\mathbf{P}) \coloneqq \Rcal^\CeBH(\mathbf{E}^{\BY(\alpha, K)}).
\end{align}
\begin{proposition}
    $R^{\BY} \subseteq R^{\cBY}$.
\end{proposition} The above proposition follows directly from the formulation of BY as eBH and \Cref{prop:cebh-dominates-ebh}. We also see empirically in \Cref{sec:simulations}, which gives us the rather surprising result that the BY procedure is inadmissible, even when restricted to deterministic procedures (as \citet{xu_more_powerful_2023} provided a randomized version of the BY procedure that dominated the original).

\subsection{The compound $\textCeBH$ procedure}\label{sec:compound}

\citet{ignatiadis_asymptotic_compound_2025} showed that eBH was a universal procedure for FDR control, i.e., every FDR controlling procedure can be written as the eBH procedure applied to some set of compound e-values. For \emph{compound e-values}, we let our data be $\mathbf{X} = \widetilde{\mathbf{E}} \coloneqq (\widetilde{E}_1, \dots, \widetilde{E}_K)$. Compound e-values are all nonnegative and satisfy the following condition
\begin{equation}\label{eq:compound-def}
\sum_{i \in A}\expect_A[\widetilde{E}_i] \leq K.
\end{equation}
for any $A \subseteq [K]$. In words, no matter which $A$ constitutes the true set of null hypotheses $A^*$, the expected sum of the corresponding compound e-values cannot exceed $K$.

Our results are complementary to the above universality result on eBH. Instead of reducing every FDR procedure to using eBH with compound e-values, we instead reduce it the $\textCeBH$ procedure with an e-value for each intersection hypothesis.
These two interpretations are interconnected, and we can formulate eBH applied to compound e-values as the application of the $\textCeBH$ procedure to an e-collection for the intersection hypotheses.

For any $A \subseteq [K]$, note that the~\eqref{eq:compound-def} implies that
\begin{align}
    E_A = K^{-1}\sum_{i \in A}\widetilde{E}_i \label{eq:c-int-evalue}
\end{align} is a valid e-value for $\hypset_A$.
Let the \emph{compound $\textCeBH$ procedure} be the $\textCeBH$ procedure that utilizes \eqref{eq:c-int-evalue} for its e-collection instead of \eqref{eq:avg-e-value} and be denoted as $\tilde{\rejset}^{\CeBH} \coloneqq \Rcal^{\CeBH}(\widetilde{\mathbf{E}})$.
Let $\tilde{\rejset}^{\eBH} \coloneqq \Rcal^\eBH(\widetilde{\mathbf{E}})$ now be the discovery set that results from applying the eBH procedure.
\begin{proposition}
    We have that $\tilde{\rejset}^{\eBH} \subseteq \tilde{\rejset}^{\CeBH}$.
\end{proposition}
\begin{proof}
The proof is nearly identical to that of the proof of \Cref{prop:cebh-dominates-ebh}, albeit using the e-values defined in \eqref{eq:c-int-evalue}.
\begin{align}
    E_A = K^{-1}\sum\limits_{i \in A} \widetilde{E}_i
        \geq
        K^{-1} |A\cap \tilde{\rejset}| \min_{i \in A \cap \tilde{\rejset}} \widetilde{E}_i
        \geq
        K^{-1}|A\cap \tilde{\rejset}|\min_{i \in \tilde{\rejset}} \widetilde{E}_i.
    \label{eq:c-lb-ineq}
\end{align}
The first inequality is lower bounding each term in the sum and also dropping some positive terms. The second inequality is by strictly expanding the set which we take a minimum over.

We have the following result for $\tilde{\rejset}^{\eBH}$:
\begin{align}
E_A
\geq
 K^{-1} |A\cap \tilde{\rejset}^{\eBH}|\min_{i \in \tilde{\rejset}^{\eBH}} \widetilde{E}_i.
\geq
\frac{|A \cap \tilde{\rejset}^{\eBH}|}{\alpha |\tilde{\rejset}^{\eBH}|}
=
\frac{\FDP_A(\tilde{\rejset}^{\eBH})}{\alpha}.\label{eq:c-ineq-avg-lb}
\end{align}
The first inequality is by \eqref{eq:c-lb-ineq}, and the second inequality is by the definition of $\Rcal^\eBH$. \eqref{eq:c-ineq-avg-lb} implies that $\tilde{\rejset}^\eBH \in \overline{\mathcal{C}}$, and consequently we get that $\tilde{\rejset}^\eBH = \rejset_{k^\eBH} \subseteq \rejset_{k^\CeBH} = \tilde{\rejset}^\CeBH$. Thus, we have shown our desired result.

\end{proof}

We conjecture that the compound $\textCeBH$ procedure and standard compound eBH procedure are actually equivalent, as we have observed simulations where this is the case, and we will leave a rigorous investigation to future work.

\section{A closure principle for FDR control}\label{sec:closure}
Now, we will argue that \emph{any} FDR controlling procedure (i.e., including procedures that take inputs that are not e-values) can be improved by or reinterpreted as a procedure based on deriving e-values for every intersection hypotheses $\hypset_A$ for each $A \subseteq [K]$.
Let $\Rcal$ be a \emph{FDR controlling procedure} (with discovery set $R = \mathcal{R}(\mathbf{X})$) if it satisfies
\begin{align}
    \sup_{A \subseteq 2^{[K]}} \expect_{A}\left[\FDP_A(R)\right] \leq \alpha. \label{eq:fdr-control}
\end{align}

In this section, we now let $E_A$ denote any e-value for the intersection hypothesis $\hypset_A$.
Now, consider any \emph{e-closed} procedure $\overline{\Rcal}$, which is any procedure that produces a discovery set that satisfies $\overline{R} \in \overline{\mathcal{C}}$ for the resulting candidate set $\overline{\mathcal{C}}$ in~\eqref{eq:candidates}.

\begin{theorem}\label{prop:e-closed-fdr}
    Any e-closed procedure $\overline{\Rcal}$ (with corresponding discovery set $\overline{R} = \overline{\mathcal{R}}(\mathbf{X})$) which uses
\begin{align}
    E_A = \frac{\FDP_A(\rejset)}{\alpha} \label{eq:standard-ec-evalue}.
\end{align} to form a e-collection $(E_A)_{A \subseteq [K]}$ ensures that in $\rejset \in \overline{\mathcal{C}}$. Further, the dominant e-closed procedure $\overline{\Rcal}$ is equivalent to $\Rcal$, i.e., $\overline{R} = R$.

Thus, any FDR controlling procedure can be represented as a e-closed procedure.
\end{theorem}
\begin{proof}
Note that $E_A$ is a valid e-value for $\hypset_A$, since $\expect_A[\FDP_A(\rejset)] \leq \alpha$ by $\rejset$ being a FDR controlling procedure. Furthermore, we can see that
\begin{align}
    \max_{A \subseteq 2^{[K]}} \frac{\FDP_A(\rejset)}{E_A} =  \max_{A \subseteq [K]} \frac{\FDP_A(\rejset)}{\FDP_A(\rejset) / \alpha}  = \alpha,
\end{align}
when we treat $0 / 0 = 1$ in the above division. Thus, $\rejset \in \overline{\mathcal{C}}$. As a result, there exists a e-closed procedure which can always choose $\overline{R} \in \overline{\mathcal{C}}$. On the other hand, for every $R' \subseteq [K]$ with $|R'| > |R|$, we have any $A = R \cap R' \cup A'$ where $\emptyset \subset A' \subseteq R' \setminus R$ has $E_A = \FDP_A(R) / \alpha < \FDP_A(R') / \alpha$. Similarly, if $0 < |R'| < |R|$, we can choose $A' \in R \setminus R'$ instead and note the same inequality. Thus, $R$ and $\emptyset$ are the only sets in $\overline{\mathcal{C}}$. As a result, letting the e-closed procedure always outputting $R$ is the dominant procedure.

\end{proof}

\paragraph{Closure for post-hoc valid procedures.} We also formulate an alternative way to form an e-collection for the e-closed procedure from procedures with post-hoc FDP control (in the sense of \eqref{eq:post-hoc-valid}). Here, let $\mathcal{D}$ be a procedure that produces a family of discovery sets $(R_\beta)_{\beta \in [0, 1]}$, and $\mathcal{D}$ is \emph{post-hoc FDP controlling procedure} if it satisfies
\begin{align}
    \expect\left[\sup_{\beta \in [0, 1]}\frac{\FDP_{A^*}(R_\beta)}{\beta}\right] \leq 1.
\end{align} We can define the following e-values for each $A \subseteq [K]$:
\begin{align}
    \bar{E}_A = \sup_{\beta \in [0, 1]} \frac{\FDP_{A}(R_\beta)}{\beta}. \label{eq:post-hoc-closed-evalue}
\end{align}
\begin{theorem}
    For an arbitrary post-hoc procedure $\mathcal{D}$, the e-closed procedure $\overline{\Rcal}$ with e-collection $(\bar{E}_A)_{A \subseteq [K]}$ that ouputs $\overline{R}$ recovers or improves upon the FDR controlling procedure $\Rcal$ where $R = \Rcal(\mathbf{X}) = \mathcal{D}(\mathbf{X})_\alpha$ for a specific $\alpha \in [0, 1]$.
\end{theorem}
\begin{proof}
    The theorem follows from $(\bar{E}_A)_{A \subseteq[K]}$ being a valid e-collection due to $\mathcal{D}$ having post-hoc FDP control. Further, we note that $\bar{E}_A \geq E_A$ (as defined in \eqref{eq:standard-ec-evalue}) since $E_A$ is included in the supremum taken in $\bar{E}_A$. As a result, the resulting e-closed procedure $\overline{R}$ provides FDR control by outputting a discovery set in $\overline{R}\in \Cclosed$ s.t.\ $\overline{R} \supseteq R$ --- this set always exists since $R \in \Cclosed$. Thus, we have shown our desired result.
\end{proof}
Let us briefly compare the construction of the e-collection in \eqref{eq:post-hoc-closed-evalue} to the one for \eqref{eq:standard-ec-evalue}. Notably, in \eqref{eq:post-hoc-closed-evalue}, the e-values do not depend upon $\alpha$. Further if eBH is plugged into the e-closed procedure applied to e-values from \eqref{eq:standard-ec-evalue}, it cannot be improved. However, every e-value in \eqref{eq:post-hoc-closed-evalue} is at least as large, and typically larger, than the corresponding e-value defined in \eqref{eq:standard-ec-evalue} (as the latter is included in the supremum of the former). As a result, typically, taking the closure actually can increase the power over the original procedure. Further, e-collection of \eqref{eq:standard-ec-evalue} sets $E_i = 0$ for $i \not\in R$. On the other hand, the e-collection of \eqref{eq:post-hoc-closed-evalue} constructs $E_i$ s.t.\ $E_i > 0$ is possible for $i\not\in R$, and they are usually positive values. Thus, the e-closed procedure derived from \eqref{eq:post-hoc-closed-evalue} could improve over the FDR controlling procedure $\Rcal$ used to derive the e-collection.

\section{Closed testing for general error metrics}\label{sec:general-closure}
Suppose we have a general error metric $\F=(\F_A)_{A\subseteq [K]}$, which means that $\F_{A}$ is a (possibly random) function that maps any rejection set $R$ to some nonnegative real number and we aim to find an $R$ such that
\begin{align}\sup_{A \subseteq 2^{[K]}} \expect_{A}\left[\F_A(R)\right] \leq \alpha.\label{eq:general_control}\end{align}
For example, in case of FDR control, we have $\F_A=\text{FDP}_A$. However, there are many other error rates this general setup includes:
\begin{itemize}
    \item k-familywise error rate (k-FWER): $\F_A(R)=\ind\{|A\cap R|\geq k\}$.
    \item Per-family wise error rate (PFER): $\F_A(R)=|A\cap R|$.
    \item False discovery exceedance (FDX): $\F_A(R)=\ind\{\FDP_A(R)> \gamma\}$ for some $\gamma\in (0,1)$.
\end{itemize}

Analogous as before, given an e-value $E_A$ for each $A\subseteq [K]$, define the set of candidate discovery sets for the error metric $\F$ by
\begin{align}
\label{eq:candidates_general}
\overline{\mathcal{C}}(F)\coloneqq \left\{\rejset \subseteq [K]:
E_A \geq \frac{\F_A(\rejset)}{\alpha} \quad \forall {A \subseteq [K]}
\right\}.
\end{align}

The validity and admissibility of this general closure principle follows exactly as in Proposition~\ref{prop:cebh_FDR} and Theorem~\ref{prop:e-closed-fdr}, respectively.

\begin{theorem}\label{theo:CP_general}
    Any set $R\in \overline{\mathcal{C}}(F)$ satisfies \eqref{eq:general_control}. Furthermore, for any rejection set $R$ and $E_A=F_A(R) / \alpha$, we have $R\in \overline{\mathcal{C}}(F)$. Hence, every procedure can be obtained as a closed procedure.
\end{theorem}

This general formulation has the further advantage that it allows the post-hoc choice of the error metric. Let $\mathcal{F}$ be the set of all error metrics under consideration and suppose for every $\F\in \mathcal{F}$, we pick a specific closed procedure $\overline{\rejset}_F\in \overline{\mathcal{C}}(F)$  (usually the largest set among all candidates). Then, we have
\begin{align}\sup_{A \subseteq 2^{[K]}} \expect_{A}\left[\sup_{\F\in \mathcal{F}} \F_A(\overline{\rejset}_F)\right] \leq \sup_{A \subseteq 2^{[K]}} \expect_{A}\left[E_A\right] \alpha \leq \alpha.\end{align}
Hence, we can choose the error metric and thus the corresponding rejection set $\overline{\rejset}_F$ post-hoc.

\begin{example}[Deciding between e-Holm and $\CeBH$ post-hoc\label{example:Holm}]
      Let $E_A=|A|^{-1} \sum_{i \in A} E_i$ be defined as in Section~\ref{sec:closed-ebh}. We know that in case of $\F_A=\text{FDP}_A$, the resulting closed procedure is given by $\CeBH$. Now suppose we seek for FWER control and thus have $\F_A(R)=\ind\{|A\cap R|> 0\}$. Then,
      \begin{align}\overline{\mathcal{C}}(F)&= \left\{\rejset \subseteq [K]:
E_A \geq \frac{\F_A(\rejset)}{\alpha} \quad \forall {A \subseteq [K]}
\right\}\\
&=\left\{\rejset \subseteq [K]:
E_A \geq \frac{1}{\alpha} \text{ for all } A \text{ with } A\cap R\neq \emptyset
\right\}.\end{align}
Therefore, the largest rejection set $\overline{\rejset}\in \overline{\mathcal{C}}(F)$ contains all indices $i\in [K]$ such that $E_A\geq 1/\alpha$ for all $A$ with $i\in A$. This is the e-Holm procedure \citep{vovk_e-values_calibration_2021, hartog_family-wise_error_2025}. Consequently, we could decide based on the data if we want to control the FDR with the $\CeBH$ procedure or to control the FWER with the e-Holm procedure without violating the guarantee.
  \end{example}

\begin{remark}
    Our general closure principle recovers the existing closure principles for FWER control \citep{marcus_closed_testing_1976} and simultaneous FDP control \citep{goeman2011multiple} by choosing $E_A=\phi_A/\alpha$ for some level-$\alpha$ intersection tests $\phi_A$, $A\subseteq [K]$.

    To see this, note that in the FWER case, $\F_A(R)=\ind\{|A\cap R|>0\}$, the largest set $\overline{\rejset}\in \overline{\mathcal{C}}(F)$ contains all indices $i\in [K]$ with $E_A\geq \alpha^{-1} \Leftrightarrow \phi_A=1$ (see also Example~\ref{example:Holm}), which is precisely the closure principle for FWER control \citep{marcus_closed_testing_1976}.

    For simultaneous FDP control,  fix a candidate rejection set $R\subseteq [K]$,
 set $F_A^{d}(R)=\ind \{|R\setminus A| <d\}$ and define
    \begin{align}
    \boldsymbol{d}(R)&\coloneqq\max\{d\in \{0,\ldots,K\}: \phi_A\geq F_A^{d}(R) \quad  \forall A\subseteq [K]\} \\
    &=\min\{|R\setminus A|:  A\subseteq [K], \phi_A=0\},
    \end{align}
    then $\boldsymbol{d}(R)$ is the closed procedure for simultaneous true discovery guarantee, which is equivalent to simultaneous FDP control \citep{goeman2021only}. By the simultaneous control of our closure principle over all error rates $F^d$, $d\in \{0,\ldots,K\}$, and all rejection sets $R\in \overline{\mathcal{C}}(F^d)$, we obtain the simultaneous true discovery of $\boldsymbol{d}$ using our general closure principle
    \begin{align}
       & \mathbb{E}_{A^*}\left[\sup_{d \in \{0,\ldots,K\}} \sup_{R\in \overline{\mathcal{C}}(F^d)} F_{A^*}^d(R)\right] \leq \alpha \\
        \implies \ &\mathbb{P}_{A^*}(\exists R\subseteq [K]: \boldsymbol{d}(R)>|R\setminus A^*|)\leq \alpha.
    \end{align}
\end{remark}

\begin{remark}
The offline multiple testing problem can be viewed as a composite generalized Neyman-Pearson (GNP) problem \citep{grunwald_neyman-pearson_e-values_2024} for different loss functions.
In particular, we can view a discovery procedure, $\Rcal$ as the maximally compatible decision rule for the multiple testing problem w.r.t.\ the e-collection derived in \Cref{prop:e-closed-fdr} for FDR control. However, admissibility results for maximally compatible decision rules from \citet{grunwald_neyman-pearson_e-values_2024} do not directly apply to our setting since their assumption on action ``richness'' (continuity of loss --- for FDP or other discrete multiple testing error metrics  --- w.r.t.\ the action space) is violated in multiple testing, as there are only a finite set of values that the FDP can take.
\end{remark}

\section{Extensions}\label{sec:extensions}

We discuss some extensions to the $\textCeBH$ procedure. These extensions primarily highlight the fact that the advantages of e-value procedures are retained for our procedure as well.

\paragraph{Randomization.} \citet{xu_more_powerful_2023} introduced stochastic rounding of e-values to use randomization to improve e-value based hypothesis tests as well as the eBH procedure. We can apply stochastic rounding individually to each intersection e-value $E_A$ for $A \subseteq [K]$, as follows.

Recall that $R_{[:k]}$ denotes the set of hypotheses corresponding to the $k$ largest e-values, and solve for $\hat k:= k^\CeBH$ from~\eqref{eq:k-cebh}. Then, define data-dependent test levels
\begin{align}
\hat\alpha_{A, k} \coloneqq \left(\frac{\FDP_A(R_{[:k]})}{\alpha}\right)^{-1}
\end{align}
for each $A \subseteq [K], k \in [K]$. Finally, define
\begin{align}
S_{A}(E_A) &= \begin{cases}
    \hat\alpha_{A, \hat k}^{-1} & \text{ if }\hat\alpha_{A, \hat k + 1} > \hat\alpha_{A, \hat k}\text{ or }U > \frac{E_A - \hat\alpha^{-1}_{A, \hat k}}{\hat\alpha^{-1}_{A, \hat k + 1} - \hat\alpha^{-1}_{A, \hat k}}\\
    \hat\alpha_{A, \hat k + 1}^{-1} & \text{ otherwise}
\end{cases}.
\end{align}
One can then check that $S_A(E_A)$ is an e-value where $U$ is uniform random variable on $[0, 1]$ that is independent of everything else.

Note that $S_A(E_A) \geq \hat\alpha^{-1}_{A, k^\CeBH}$ almost surely, which means that $\textCeBH$ using the e-collection $(S_A(E_A))_{A \subseteq [K]}$ will never make fewer discoveries. Further, if $\min_{A \subseteq [K]} E_A - \hat\alpha^{-1}_{A, k^{\CeBH}} > 0$, there is positive probability that $k^\CeBH + 1$ discoveries can be made. In such a case, the randomized $\textCeBH$ procedure is strictly more powerful than the $\textCeBH$ procedure.

\paragraph{Boosting.} \citet{wang_false_discovery_2022} observed that one can increase the power of an e-value when testing against a finite set of rejection thresholds. In the most general setting of assuming arbitrary dependence, if one knows the marginal distribution of an e-value under the null, one can calculate a boosting factor $b$ for an e-value $E$ as follows $$b \coloneqq \max\ \{b' \in [1, \infty): \expect[T(b' E)] \leq 1\}$$ where the expectation is taken under the known null distribution. Here, $T$ is a truncation function that takes the largest element in a predetermined set that is less than or equal to the input.

In our setting, one simple way to boost e-values is to define a different truncation function $T_A$ for each $A \subseteq [K]$ as follows: $$T_A(x) \coloneqq \max\ \{r / (\alpha k) : r / (\alpha k) \leq x, k \in [K], r \in [k \wedge |A|]\} \cup \{0\},$$
where $a \wedge b$ is the minimum of $a$ and $b$. This is a novel truncation function unlike ones used in boosting for the eBH procedure, and is derived from restricting the support of $E_A$ to only the values $\FDP_A(R) / \alpha$ can take. Let $b_A$ be the boosting factor for $E_A$ with truncation function $T_A$ for each $A \subseteq [K]$. Then, one can notice that $T_A(b_A E_A)$ is a valid e-value and is at least as powerful as $E_A$ in the sense that $\FDP_A(R) / \alpha \leq E_A$ implies $\FDP_A(R) / \alpha \leq T_A(b_A E_A)$ for all $R \subseteq [K]$.

Using the above boosted e-values immediately yields that the boosted $\textCeBH$ procedure that dominates the original.
Other boosting techniques may be possible given more dependence assumptions or knowledge about the joint distribution of the entire e-collection $(E_A)_{A \subseteq [K]}$.

\paragraph{Alternative e-merging functions.} While we used the average e-merging function in our formulation of the $\textCeBH$ procedure, one can use other e-merging functions when stronger dependence assumptions, i.e., independence or sequential dependence, are made about the e-values.
For example, if we assume all null e-values are independent, one can let the e-values for each intersection hypothesis be the product of the base hypothesis e-values, i.e.,
\begin{align}
    E_A = \prod\limits_{i \in A} E_i.
\end{align} One can then proceed to define a $\textCeBH$ procedure for the product e-merging function in a similar fashion as we did in \Cref{sec:closed-ebh}, along with a corresponding algorithm, by using the logarithm to turn the products into sums when computing e-values for the intersection hypotheses.

\section{Numerical simulations}\label{sec:simulations}
\begin{figure}[htbp]
    \begin{subfigure}{0.29\textwidth}
        \includegraphics[trim=0 0 3.2cm 0,clip,width=\textwidth]{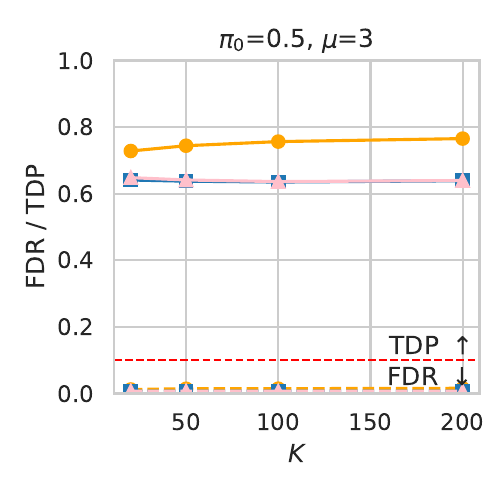}
    \end{subfigure}\begin{subfigure}{0.29\textwidth}
        \includegraphics[trim=0 0 3.2cm 0,clip,width=\textwidth]{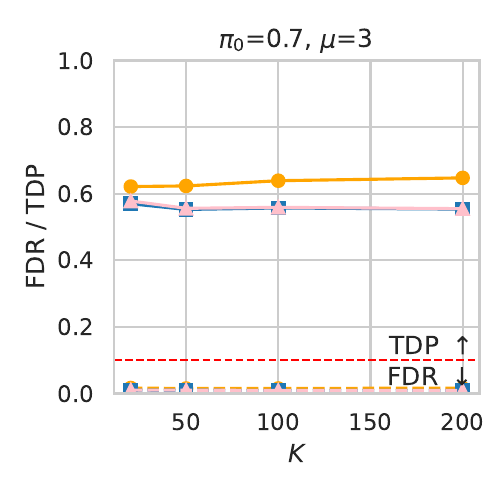}
    \end{subfigure}\begin{subfigure}{0.29\textwidth}
        \includegraphics[trim=0 0 3.2cm 0,clip,width=\textwidth]{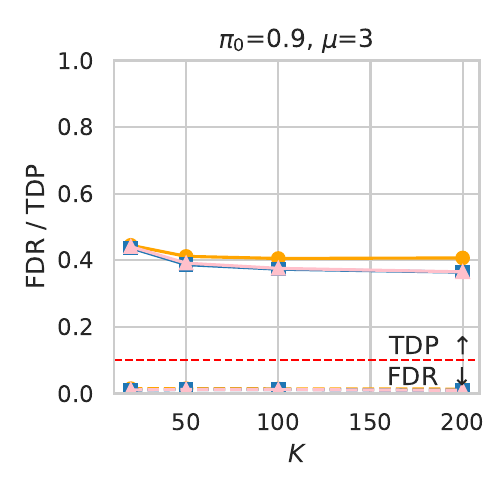}
    \end{subfigure}\hfill\begin{subfigure}{0.13\textwidth}
        \includegraphics[trim=8.2cm 0 0 0,clip,width=\textwidth]{figures/simple/null0.9_signal3.pdf}
    \end{subfigure}
    \caption{Plot of FDR and TPR for different null proportions $\pi_0$ and alternative mean $\mu = 3$. All differences are significant since the closed eBH ($\overline{\eBH}$) procedure dominates the eBH and minimally adaptive eBH (eBHm) procedures.}
    \label{fig:gaussian-sim}
\end{figure}

We perform simulations\footnote{Code at \url{https://github.com/neilzxu/closed-ebh}} for independent Gaussians where we let $\pi_0 \coloneqq |A^*| / K$ be the null proportion, $\mu$ be the signal of the alternatives and let $X_i \sim \mathcal{N}(0, 1)$ for $i \in A^*$ and $X_i \sim \mathcal{N}(\mu, 1)$ otherwise for $i \in [K] \setminus A^*$.
We let our e-value be $E_i \coloneqq \exp(\lambda X_i - \lambda^2 / 2)$ for $\lambda = \mu$ (the log-optimal choice).
We let $\alpha = 0.1$ in our simulations. We plot the empirical average FDR and true positive rate (TPR) $\coloneqq \expect[|(A^*)^c \cap \rejset| / |(A^*)^c|]$ of eBH, minimally adaptive eBH \citep{ignatiadis_e-values_unnormalized_2022}, and $\textCeBH$ over $n = 1000$ trials.
We plot the results in \Cref{fig:gaussian-sim}, and can see that the $\textCeBH$ procedure improves noticeably over the standard and minimally adaptive eBH procedures across all settings, while controlling the FDR below $\alpha$.

\begin{figure}[htbp]
    \begin{subfigure}{0.29\textwidth}
        \includegraphics[trim=0 0 3.2cm 0,clip,width=\textwidth]{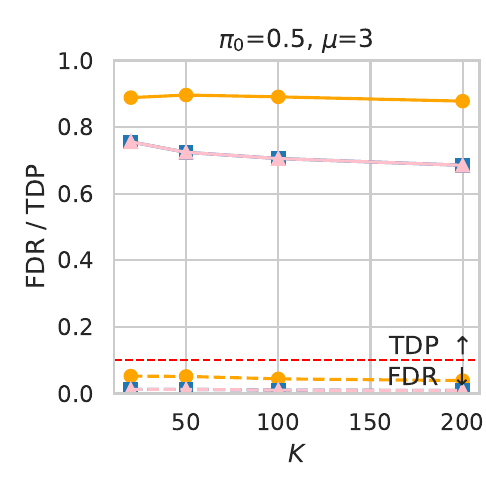}
    \end{subfigure}\begin{subfigure}{0.29\textwidth}
        \includegraphics[trim=0 0 3.2cm 0,clip,width=\textwidth]{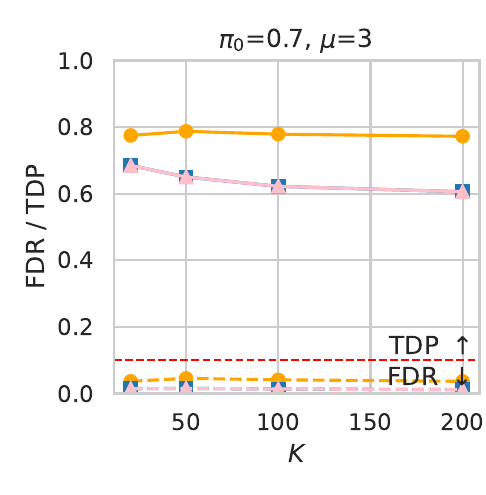}
    \end{subfigure}\begin{subfigure}{0.29\textwidth}
        \includegraphics[trim=0 0 3.2cm 0,clip,width=\textwidth]{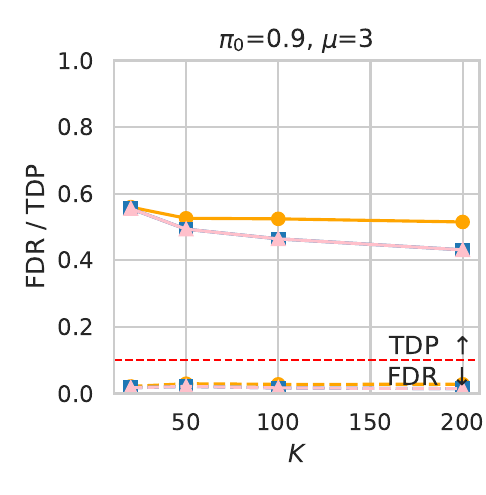}
    \end{subfigure}\hfill\begin{subfigure}{0.13\textwidth}
        \includegraphics[trim=8.2cm 0 0 0,clip,width=\textwidth]{figures/simple-BY/null0.9_signal3.pdf}
    \end{subfigure}
    \caption{Plot of FDR and TPR for different null proportions $\pi_0$ and alternative mean $\mu = 3$. All differences are significant since the closed BY ($\textcBY$) procedure dominates other variants of the BY procedure.}
    \label{fig:by-sim}
\end{figure}
We also evaluate the BY procedure, the application of minimally adaptive eBH to the e-values defined in \eqref{eq:by-calib-evalues}, and the $\textcBY$ procedure where we set $P_i \coloneqq 1 - \Phi(X_i)$ and $\Phi$ is the c.d.f.\ of the standard normal distribution. Here, we draw $(X_1, \dots, X_K)$ from a multivariate normal with the same component means as described above and with Toeplitz-like covariance matrix where $|\text{Cov}(X_i, X_j)| = \exp(-|i - j| / 10) / 5$ with the covariance being positive if $|i - j|$ is even and negative otherwise. Similarly, our results in \Cref{fig:by-sim} show that $\textcBY$ has substantial power increase over the other BY procedures.

\section{Related work}\label{sec:related-work}
We will briefly describe the two lines of research and prior art that are most relevant to our current work. Aside from these, there is a substantial literature using e-values and the eBH procedure for FDR control --- \citet{ramdas_hypothesis_testing_2025} performs a thorough overview of such developments.

\paragraph{Closed testing and e-values.} There is a vast amount of literature for FWER control and simultaneous (over all potential discovery sets) FDP bounds using closed testing --- see overviews presented in \citet[Chapters 9 and 14]{cui_handbook_multiple_2021} and \citet{henning_closed_testing_2015}. We will point out two specific works that are particularly relevant to ours since they also utilize e-values in the context of closed testing.
The first is \citet{vovk_confidence_discoveries_2023a}, who formulate a method for FWER control via the arithmetic mean e-merging function, and consequently also simultaneous FDP bounds.
\citet{hartog_family-wise_error_2025} extend the benefit of using arithmetic mean e-merging functions to graphical approaches for FWER control \citep{bretz_graphical_approach_2009}, and derive computationally efficient algorithms for e-value versions of the Fallback \citep{wiens_fallback_procedure_2005} and \citet{holm_simple_sequentially_1979} procedures.
Our work proposes new procedures the FDR metric (rather than FWER or probabilistic bounds on the FDP) which has not been connected to intersection hypothesis testing or e-merging functions before, and we show existing procedures (such as e-Holm) can be formualted using the generalized closure principle we develop in \Cref{sec:general-closure}.

\paragraph{Boosting the power of eBH} Another line of work that is pertinent is various approaches to boosting the eBH procedure so that more discoveries are made. There are generally two areas in this line: (1) boosting with no assumptions on the e-value distributions and (2) boosting e-values by using explicit knowledge about the underlying (marginal or joint) distribution.

For the first area, we have already discussed the connection between our work and the stochastic rounding framework of \citet{xu_more_powerful_2023} in \Cref{sec:extensions}, as well as the minimally adaptive eBH procedure of \citet{ignatiadis_e-values_unnormalized_2022} in \Cref{sec:closed-ebh}.

The second area includes two general types of assumptions.
\begin{itemize}
    \item \emph{Marginal assumptions:} Along with introducing the eBH procedure, \citet{wang_false_discovery_2022} also outlined how one can boost e-values if one knows the explicit marginal distribution of each e-value, and we discuss a simple variant for boosting the $\textCeBH$ procedure in \Cref{sec:extensions}. \citet{blier-wong_improved_thresholds_2024} relaxed the necessary knowledge on marginal distributions for boosting by deriving boosting thresholds for certain nonparametric classes of e-values.
    \item \emph{Conditional assumptions:} \citet{lee_boosting_e-bh_2024a,lee_full-conformal_novelty_2025} improve the power of eBH in the setting where one knows the distribution of all $K$ e-values when conditioned a sufficient statistic calculated from the data under a singular null (i.e., a specific hypothesis is null, and all other hypotheses may be either null or non-null). This can also be viewed as a framework for deriving compound e-values from e-values when their joint distribution is known. Boosting factors can be calculated for each e-value via exact knowledge of these conditional distributions.
\end{itemize}
Developments in this area are orthogonal to our contribution. These boosting methods all require additional knowledge about the underlying distributions. On the other hand, the $\textCeBH$ procedure requires no assumptions on the e-values, and can be applied with valid FDR control in all situations where the eBH procedure can be applied. Potentially, these boosting methods can be used in conjunction with $\textCeBH$ to further improve its power.

\section{Conclusion}\label{sec:conclusion}

We have shown in this manuscript that the eBH procedure is inadmissible, and can be improved through a FDR closure principle that results in the $\textCeBH$ procedure. We provide a computationally efficient way of computing the $\textCeBH$ discovery set. Further, our FDR closure principle dictates that all FDR controlling procedures rely on e-values for testing each possible intersection hypothesis of the base hypotheses. Moreover, we generalize our FDR closure principle to general error metrics, which particularly allows to choose the error metric post-hoc.
These results open many lines of future work
.

\begin{comment}
\subsection{Questions}

Some questions now are:
\begin{enumerate}
    \item Is $\rejset(\alpha)$ the same as $\rejset^\eBH(\alpha)$? Probably not, since there are ways to improve eBH (e.g., I believe you can global null test first, then eBH on $K-1$ hypotheses).
    \item What if $(E_1, \dots, E_K)$ are not arbitrarily dependent? Can we do better?

    We know all FDR controlling procedures are eBH, but we don't know what the optimal/admissible e-values are. Can we do better or if not, at least show certain kinds of boosting are admissible?
\end{enumerate}
Let's explore the second question.
\end{comment}
\bibliography{ref}
\end{document}